\documentclass[sn-mathphys,Numbered]{sn-jnl}
\usepackage{graphicx}
\usepackage{multirow}
\usepackage{amsmath,amssymb,amsfonts}
\usepackage{amsthm}
\usepackage{mathrsfs}
\usepackage[title]{appendix}
\usepackage{xcolor}
\usepackage{textcomp}
\usepackage{manyfoot}
\usepackage{booktabs}
\usepackage{algorithm}
\usepackage{algorithmicx}
\usepackage{algpseudocode}
\usepackage{listings}

\theoremstyle{thmstyleone}
\newtheorem{theorem}{Theorem}
\newtheorem{proposition}[theorem]{Proposition}
\newtheorem{lemma}[theorem]{Lemma}

\theoremstyle{thmstyletwo}

\theoremstyle{thmstylethree}
\newtheorem{definition}{Definition}
\begin{document}

\title{Graph structure of quantum mechanics}

\author[1]{\fnm{Songyi} \sur{Liu}}\email{liusongyi@buaa.edu.cn}

\author*[1]{\fnm{Yongjun} \sur{Wang}}\email{wangyj@buaa.edu.cn}

\author[1]{\fnm{Baoshan} \sur{Wang}}\email{bwang@buaa.edu.cn}

\author[1]{\fnm{Jian} \sur{Yan}}\email{jianyanmath@buaa.edu.cn}

\author[1]{\fnm{Heng} \sur{Zhou}}\email{zhouheng@buaa.edu.cn}

\affil*[1]{\orgdiv{School of Mathematical Science}, \orgname{Beihang University}, \orgaddress{ \city{Beijing}, \postcode{100191}, \country{China}}}

\abstract{The quantum mechanics is proved to admit no hidden-variable in 1960s, which means the quantum systems are contextual. Revealing the mathematical structure of quantum mechanics is a significant task. We develop the approach of partial Boolean algebra to characterize the contextuality theory with local consistency and exclusivity, and then prove that the finite dimensional quantum systems are determined by atoms using two graph structure theorems. We also generalize our work to infinite dimensional cases. Our conclusions indicate that the quantum mechanics is a graph-structured combination of multiple hidden-variable theories, and provide a precise mathematical framework for quantum contextuality.}


\maketitle

\section*{Declarations}

\bmhead{Competing interests}
The authors have no relevant financial or non-financial interests to disclose.

\section{Introduction}

In the early 20th century, the establishment of quantum mechanics spurred research into the mathematical foundation of nature. Some researchers like Einstein believe that the quantum mechanics should be characterized by hidden-variable theory\citep{Fine1990Einstein}. In other words, all the observables can be assigned values simultaneously by an unknown hidden-variable $\lambda$. After a period of controversy, hidden-variable theory was proven inapplicable to quantum mechanics in 1960s.\par

The nonexistence of hidden-variables was proved by two significant achievements: Bell nonlocality\citep{Bell1964On} and Kochen-Specker theorem\citep{Kochen1967The}. The former is verified experimentally by Clauser-Horne-Shimony-Holt (CHSH) experiment\citep{Clauser1969Proposed}, and the later is a mathematical conclusion. They paved two different paths for the exploration of quantum mechanics. Bell nonlocality has inspired the research in applications such as quantum computation, quantum communication channel and quantum cryptography\citep{Brunner2014Bell}, while Kochen-Specker theorem pioneers a field later known as the contextuality theory\citep{Adan2023Kochen}, which remains dedicated to revealing the mathematical foundation of quantum mechanics.\par

The hidden-variable theory supposes that the nature is determined by a hidden-variable space, and the random phenomenon is entirely attributed to the randomness of hidden-variables. Therefore, all the observables are essentially commeasurable. However, Bell nonlocality and Kochen-Specker theorem point out that the randomness and incompatibility are essential in quantum mechanics. A set of commeasurable observables is called a context. The hidden-variable theory exactly describes a single context, while quantum mechanics consists of multiple contexts. This accounts for the failure of hidden-variable theory.\par

A system with multiple contexts is called a contextual system, which must be described by contextuality theory. A key to reveal the mathematical foundation of quantum mechanics is to characterize the structure of contexts. Numerous methods have been introduced such as the standard quantum logic\citep{Birkhoff1936The}, property lattices\citep{Coecke2002Quantum}, Topos theory\citep{Isham1998Topos}, marginal problem\citep{Fritz2013Entropic}, sheaf theory\citep{Abramsky2011sheaf}, compatibility hypergraph\citep{Acin2015A}, exclusivity graph\citep{Adan2014Graph}, etc. Nevertheless, the characterization of structure of quantum contexts is still a tough problem\citep{Adan2023Kochen}.\par

In this article, we will prove that the contexts in quantum mechanics form a concise and elegant mathematical structure, graph. According to the conclusion, a context is just a complete graph, and the structure of quantum contexts depends only on how these contexts intersect at the vertices. Such an opinion is also hinted at within the exclusivity graph, but the precise proof cannot be completed using the graph approach alone.\par

In reality, a context corresponds to a hidden-variable theory, thus the structure of contexts is exactly the combination of multiple hidden-variable theories. In most contextuality theories, the structure of hidden-variable theory has not been given adequate consideration. However, the answer is not complicated. A hidden-variable theory is equivalent to a classical probability theory, that is, a Boolean algebra.\par

Therefore, the problem is reduced to the combination of multiple Boolean algebras. Kochen and Specker adopted this perspective, and introduced the partial Boolean algebra which proved the Kochen-Specker theorem in 1960s. It is appropriate to use partial Boolean algebra to describe quantum systems because quantum mechanics satisfies local consistency\citep{Abramsky2015Contextuality,Ramanathan2012Generalized}, also called non-signaling\citep{Popescu1994Quantum} in Bell experiments. Local consistency ensures that the probability of a given event is consistent across various contexts, allowing it to be treated as one element of the system. Therefore, a partial Boolean algebra is the combination of multiple Boolean algebras. In 2015, Kochen initiated the works of reconstructing the foundation of quantum mechanics with partial Boolean algebra\citep{Kochen2015Reconstruction}, which presents an advanced progress in the mathematical structure of quantum mechanics.\par

However, partial Boolean algebra is not ``natural" enough to reflect the graph structure of quantum mechanics, because it lacks an important property: exclusivity, also known as local orthogonality\citep{Fritz2013Local}, Specker's exclusivity principle\citep{Adan2012Specker}, etc. Exclusivity ensures that the probability sum of exclusive events is less than 1. Abramsky and Barbosa prove that exclusivity is equivalent to transitivity\citep{Abramsky2020The}, which means that the logic reasoning across different contexts cannot be carried out without exclusivity. And we will prove that the exclusivity is necessary for the graph structure of quantum mechanics.\par

In this paper, we develop the work proposed by Kochen and Specker\citep{Kochen1967The,Kochen2015Reconstruction} , and develop the conclusions from Abramsky and Barbosa\citep{Abramsky2020The}. We introduce in section \ref{sec_epBA} the basic conceptions of exclusive partial Boolean algebra ($epBA$), and prove two graph structure theorems for finite $epBA$ in section \ref{sec_thm}. In section \ref{sec_quantum}, we define "quantum system" ($QS$) to show that quantum mechanics can be depicted by $epBA$, and prove the graph structure theorems for finite $QS$. The generalization of our conclusions to infinite cases is discussed in the section \ref{sec_infinite}, and the relevant comparison between our work and the known exclusivity graph approach is provided in section \ref{sec_exgraph}.

We will show that most systems of interest are determined by atom graphs in contextuality theory, just as the systems are determined by hidden-variable space in hidden-variable theory. Table \ref{h-c-q} shows the analogue notions between hidden-variable theory, contextuality theory and quantum mechanics.

\begin{table}[h]
\renewcommand{\arraystretch}{1.5}
\begin{tabular}{|c|c|c|}
\hline
Hidden-variable theory & Contextuality theory & Quantum mechanics\\
\hline
Boolean algebra $B\in BA$  & Exclusive partial Boolean algebra $B\in epBA$  & Quantum system $Q\in QS$  \\
Hidden-variable space $\Lambda$ & Atom graph $AG(B)$ & Atom graph $AG(Q)$ \\
Hidden-variable $\lambda$ & Atom $a$ & Atom projector $\hat{P}$ \\
Probability measure $p$ & State $p$ & Quantum state $\rho$ \\
\hline
 \end{tabular}
  \caption{Analogue notions between hidden-variable theory, contextuality theory and quantum mechanics}\label{h-c-q}
\end{table}

\section{Exclusive partial Boolean algebra}\label{sec_epBA}

Hidden-variable theory is based on Boolean algebra. The well-known axiomatized definition of classical probability theory is based on $\sigma$ algebra, which is also a type of Boolean algebra. Boolean algebra characterizes the logic of classical world, so it can be utilized to describe the single context in contextuality theory. $BA$ is used to denote the collection of all Boolean algebras\par

Stone's representation theorem indicates that any Boolean algebra is a set algebra. Therefore, any finite Boolean algebra $B\in BA$ is the power-set algebra of a finite set $\Lambda$, that is, $B\cong\mathcal{P}(\Lambda)$. $\Lambda$ is exactly the hidden-variable space (also known as sample space or phase space), and $B$ is the event algebra.\par

In quantum mechanics, multiple contexts correspond to multiple Boolean algebras. Kochen and Specker introduce the partial Boolean algebra to explain how these Boolean algebras form the global system. The definition below is from \cite{Van2012Noncommutativity}.

\begin{definition}[partial Boolean algebra]
    If $B$ is a set with
    \begin{itemize}
    \item a reflexive and symmetric binary relation $\odot\subseteq B\times B$,
    \item a (total) unary operation $\lnot:\ B\rightarrow B$,
    \item two (partial) binary operations $\land,\ \lor:\ \odot\rightarrow B$,
    \item elements $0,1 \in B$,
    \end{itemize}
    satisfying that for every subset $S\subseteq B$ such that $\forall a,\ b\in S,\ a\odot b$, there exists a Boolean subalgebra $C\subseteq B$ determined by $(C,\land,\lor,\lnot,0,1)$ and $S\subseteq C$, then $B$ is called a \textbf{partial Boolean algebra}, written by $(B,\odot)$, or $(B,\odot;\land,\lor,\lnot,0,1)$ for details.\par
    We use $pBA$ to denote the collection of all partial Boolean algebras.
\end{definition}

The abbreviation $pBA$ is adopted from \cite{Abramsky2020The}, where $pBA$ represents the category of partial Boolean algebras. The relation $\odot$ represents the compatibility. $a\odot b$ if and only if $a,b$ belong to a Boolean subalgebra. Therefore, a Boolean subalgebra of $B$ is called a context, and a maximal Boolean subalgebra is called a maximal context.\par

Partial Boolean algebra permits an element to belongs to different Boolean subalgebra, which has achieved great success, because it meets the local consistency of quantum mechanics. The probability of an event is consistent across different contexts. For example, in the CHSH experiment\citep{Clauser1969Proposed}, Alice measures incompatible observables $a$ and $a'$ of a particle. Holding another particle, Bob measures $b$ and $b'$. The event $a=1$ happens with the same probability regardless of the measurement done by Bob. In other words, $p(a=1)=p(a,b=1,1)+p(a,b=1,-1)=p(a,b'=1,1)+p(a,b'=1,-1)$. The property is called non-signaling in Bell scenario. In general scenario, it is called local consistency.\par

However, a partial Boolean algebra may be really not natural. To explain it, we give the definition of partial order\citep{Abramsky2020The}.\par

\begin{definition}
If $B\in pBA$, $a,b\in B$, then $a\leq b$ is defined by $a\odot b$ and $a\land b=a$.
\end{definition}

Therefore, the relation $\leq$ is defined by the partial order on the Boolean subalgebras, which depicts the logical reasoning in a context. The logical reasoning across multiple contexts is not permitted unless they are contained in a larger context.\par

Next we give an example of "unnatural" partial Boolean algebra. Consider $B_1$ and $B_2\in pBA$. They are similar to each other and illustrated in the Fig.\ref{pB1} and Fig.\ref{pB2}, where the lines represent the partial orders incompletely.

\begin{figure}[H]
      \begin{minipage}[t]{0.5\linewidth}
          \centering
          \includegraphics[width=0.9\linewidth]{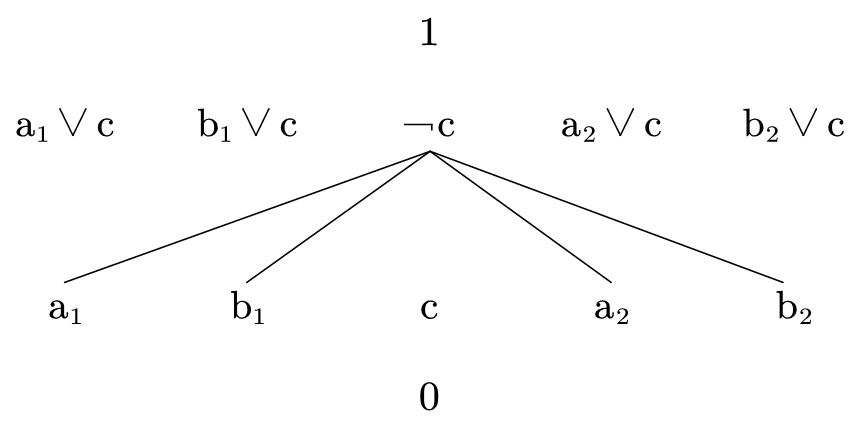}
          \caption{$B_1$}\label{pB1}
      \end{minipage}
      \begin{minipage}[t]{0.5\linewidth}
          \includegraphics[width=0.9\linewidth]{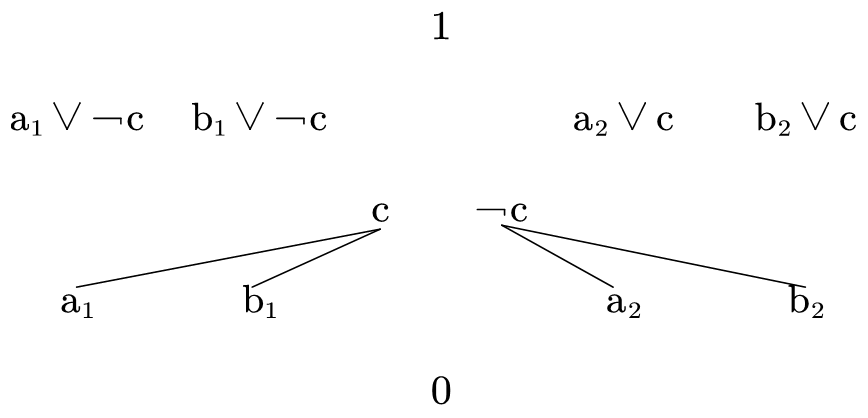}
          \caption{$B_2$}\label{pB2}
      \end{minipage}
\end{figure}

$B_1$ and $B_2$ have the same elements. $B_2$ contains two maximal Boolean subalgebras $C_1$ and $C_2$, which are isomorphic to $\mathcal{P}(\{a_1,b_1,\neg c\})$ and $\mathcal{P}(\{a_2,b_2,c\}$. It induces that the atom $\neg c$ of $C_1$ is not atom in another context $C_2$, and the atom $c$ of $C_2$ is not atom in $C_1$. Furthermore, we have $a_1\leq c$ and $c\leq a_2\lor c$, but no $a_1\leq a_2\lor c$ because $a_1$ and $a_2\lor c$ are not in a same Boolean subalgebra. In other words, $B_2$ does not satisfy transitivity, while $B_1$ does not present such issue.\par

The scenario of $B_2$ does not occur in quantum mechanics because of the exclusivity\citep{Fritz2013Local,Adan2012Specker}. For any quantum state, the probability sum of exclusive events is not more than 1. Abramsky et al. extends exclusivity from quantum states to partial Boolean algebras\citep{Abramsky2020The}. We introduces their conclusions below.

\begin{definition}
 Let $B\in pBA$. \par
 $a,b\in B$ are said to be exclusive, written $a\bot b$, if there exists an element $c\in B$ such that $a\leq c$ and $b\leq\neg c$. \par
 $B$ is said to satisfy \textbf{Logical Exclusivity Principle (LEP)} or to be \textbf{exclusive} if $\bot\subseteq\odot$.\par
 $B$ is said to be transitive if $a\leq b$ and $b\leq c$, then $a\leq c$ for any $a,b,c\in B$.
\end{definition}

\begin{theorem}[Abramsky and Barbosa\citep{Abramsky2020The}]
If $B\in pBA$, then $B$ is exclusive iff $B$ is transitive.
\end{theorem}

The conclusion indicates that exclusivity is equivalent to transitivity. We adopt the abbreviation  $epBA$ from \cite{Abramsky2020The} to denote the collection of all partial Boolean algebras satisfying LEP, and call them "exclusive partial Boolean algebra" for simplicity. $epBA$ describes the "supraquantum" systems with local consistency and exclusivity, which are not only properties of quantum mechanics, but also requirements for any contextual system capable of normal logic reasoning.\par

In addition to $epBA$, \cite{Abramsky2020The} also introduces the definition of probability distribution on contextual systems.

\begin{definition}
If $B\in pBA$,  then a \textbf{state} on $B$ is defined by a map $p:B\rightarrow[0,\ 1]$ such that
 \begin{itemize}
    \item $p(0)=0$.
    \item $p(\neg x)=1-p(x)$.
    \item for all $x,y\in B$ with $x\odot y$,$\ p(x\lor y)+p(x\land y)=p(x)+p(y)$.
 \end{itemize}
A state is called a \textbf{0-1 state} if its range is $\{0,1\}$. Use $s(B)$ to denote the states set on $B$.\label{def_state}
\end{definition}

A 0-1 state is exactly a homomorphism from $B$ to $\{0,1\}$, that is, a truth-values assignment. Definition \ref{def_state} only requires finite additivity because $pBA$ only demands closure under finite unions, and we will extend it to infinite cases later in this paper. \par

Although not yet thoroughly investigated, $epBA$ possesses more elegant and useful properties. For example, if $B$ is a finite $epBA$, then its atoms are exactly the atoms of all the maximal Boolean subalgebras, and $B$ is determined by its atoms just like a finite Boolean algebra! These conclusions will be proved in the following sections.

\section{Graph Structure Theorems of Finite epBA}\label{sec_thm}

If $B$ is a finite Boolean algebra, whose atom set is $A(B)$, then $B=\mathcal{P}(A(B))$. The theorem is fundamental for hidden-variable theory, where $A(B)$ is the set of hidden-variables. To some extent, finite $epBA$ also satisfies this theorem. Firstly, we introduce the atoms of $pBA$.

\begin{definition}
     Let $B\in pBA$, $a\in B$ and $a\neq 0$. $a$ is called an \textbf{atom} of $B$ if for each $x\in B$, $x\leq a$ implies $x=0$ or $x=a$. Use $A(B)$ to denote the atoms set of $B$.
\end{definition}

For example, the atoms of $B_1$ in Fig.\ref{pB1} are $a_1,b_1,c,a_2$ and $b_2$, while the ones of $B_2$ in Fig.\ref{pB2} are $a_1,b_1,a_2$ and $b_2$. $c$ is the atom of the maximal Boolean subalgebra of $B_2$, but it is not the atom of $B_2$. This phenomenon destroys the normal logic reasoning on $B_2$. Fortunately, the lemma below shows that exclusivity can exclude the phenomenon.

\begin{lemma}\label{lem1_finite}
If $B$ is a finite $epBA$, and $C\subseteq B$ is a maximal Boolean subalgebra, then $A(C)\subseteq A(B)$.
\end{lemma}
\begin{proof}
If $C=\{0,1\}$, then $B=\{0,1\}$ because $C$ is maximal. $1$ is the only atom of both $B$ and $C$. We suppose $C$ has at least two atoms below.\par
Let $c\in A(C)$. If $c\notin A(B)$, there exists $b\in B$ such that $b\leq c$, $b\neq c$ and $b\neq 0$, so $b\odot c$ and $b\notin C$. For any other atom $c'\neq c$ of $C$, we have $c'\leq\lnot c$, so $b\bot c'$ and then $b\odot c'$ because $B$ is exclusive. Therefore, $\{b\}\cup A(C)$ is contained in a Boolean subalgebra of $B$, which contradicts that $C$ is maximal. In conclusion, $c\in A(B)$.
\end{proof}

Conversely, it is obvious that the atoms of $B$ are atoms of its maximal Boolean subalgebras. Therefore, if $B\in epBA$, the maximal contexts of $B$ have hidden-variables of equal status. The lemma implies that it is possible to reconstruct $B$ solely from its atoms. We define "atom graph" of $pBA$ to demonstrate the conjecture.

\begin{definition}
If $B\in pBA$, the \textbf{atom graph} of $B$, written $AG(B)$, is defined by a graph with vertice set $A(B)$ such that $a_1,a_2\in A(B)$ are adjacent iff $a_1\odot a_2$ and $a_1\neq a_2$ .
\end{definition}

An atom graph is a simple graph. If $B$ is a finite Boolean algebra, then $AG(B)$ is a complete graph. The atom graphs of $B_1$ in Fig.\ref{pB1} and $B_2$ in Fig.\ref{pB2} are show in Fig.\ref{AG(pB1)} and \ref{AG(pB2)} as below.

\begin{figure}[H]
      \begin{minipage}[t]{0.5\linewidth}
          \centering
          \includegraphics[width=0.7\linewidth]{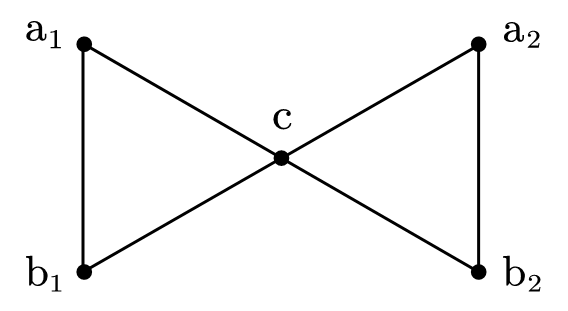}
          \caption{$AG(B_1)$}\label{AG(pB1)}
      \end{minipage}
      \begin{minipage}[t]{0.5\linewidth}
          \includegraphics[width=0.7\linewidth]{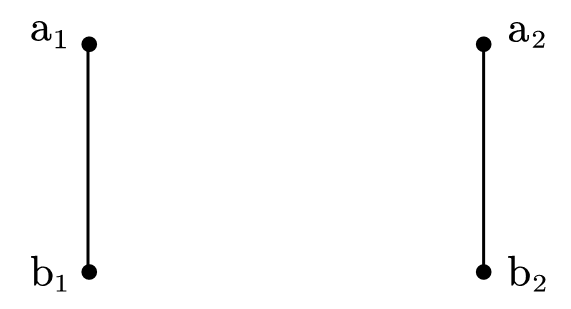}
          \caption{$AG(B_2)$}\label{AG(pB2)}
      \end{minipage}
\end{figure}

$AG(B_1)$ includes all the atoms of maximal contexts of $B_1$, while $AG(B_2)$ does not. It means that $AG(B_2)$ cannot reconstruct $B_2$. In fact, $AG(B_2)$ is also the atom graph of $B_2'=\{0,a_1,b_1,a_2,b_2,1\}$ with maximal Boolean subalgebras $\{0,a_1,b_1,1\}$ and $\{0,a_2,b_2,1\}$, and obviously $B_2'\in epBA$. It is natural to consider when a contextual system can be reconstructed from its atom graph. The theorem \ref{thm1} below answers this question.

\begin{theorem}\label{thm1}
    If $B_1$ and $B_2$ are finite $epBA$. then $B_1\cong B_2$ iff $AG(B_1)\cong AG(B_2)$.
\end{theorem}

The details of proof is given in the Appendix \ref{A}.\par

Theorem \ref{thm1} indicates that there exists a one-to-one correspondence between finite $epBA$ and their atom graphs, and it is easy to prove that $B_1=B_2$ iff $AG(B_1)=AG(B_2)$. Therefore, even though a system admits no hidden-variables, it can be reconstructed by the hidden-variables of maximal contexts, provided that it satisfies local consistency and exclusivity. And the structure of atom graph determines the structure of contextual system.\par

Next we consider whether the states on an $epBA$ can be reconstructed by its atoms.

\begin{definition}
If $G$ is a simple graph whose cliques have finite sizes, a \textbf{state} on $G$ is defined by a map $p:V(G)\rightarrow[0,\ 1]$ such that for each maximal clique $C$ of $G$, $\sum_{v\in C}p(v)=1$. Use $s(G)$ to denote the states set on $G$.
\end{definition}

For $B\in epBA$, a maximal clique of $AG(B)$ is a maximal complete subgraph, which corresponds to a maximal Boolean subalgebra. Therefore, the states on $AG(B)$ describe the probability distributions on $A(B)$. The state on a finite Boolean algebra can be reconstructed by the probability distributions on its atoms set\citep{Abramsky2020The}, and we can generalize the conclusion to finite $epBA$.

\begin{theorem}\label{thm2}
If $B$ is a finite $epBA$, then $s(B)\cong s(AG(B))$.
\end{theorem}

The proof is given in Appendix \ref{A}.\par

For details, if $p\in s(B)$, then $p|_{A(B)}$ is a state on $AG(B)$. Conversely, if $p'$ is a state on $AG(B)$, there exists a unique state $p\in s(B)$ such that $p'=p|_{A(B)}$. Therefore, the states on $AG(B)$ determine the states on $B$, just like the case in hidden-variable theory.\par

We call the theorems \ref{thm1} and \ref{thm2} the graph structure theorems of finite $epBA$. They reveal that contextual systems with local consistency and exclusivity are determined by their atom graphs. Therefore, atom graph can be considered the mathematical foundation of quantum mechanics, just as hidden-variable space underlies the classical world. We will introduce how to describe quantum mechanics using $epBA$ in the next section.

\section{Quantum system}\label{sec_quantum}
Quantum mechanics points out that the world is essentially random, thus all the physical phenomena can be depicted by "event" and "the probability of event". An event refers to the outcome of a single measurement, or equivalently, the value of an observable. If the considered Hilbert space is $\mathcal{H}$, then an observable $A$ is depicted by a bounded self-adjoint operator $\hat{A}$ onto $\mathcal{H}$. An event is a proposition like $A\in\Delta$ ($\Delta$ is a Borel set of $\mathbb{R}$), which is depicted by a projector $\hat{P}$ onto the corresponding eigenspace of $\hat{A}$. Therefore, the event algebra of quantum mechanics is the projector algebra. In 1936, Birkhoff and Von Neumann introduced the standard quantum logic based on this perspective\citep{Birkhoff1936The}.\par

Let $P(\mathcal{H})$ denote the set of projectors onto $\mathcal{H}$. If $\hat{P}_1,\ \hat{P}_2$ are projectors onto $S_1,\ S_2$, $\hat{P}_1\land\hat{P}_2$ is defined to be the projector onto $S_1\cap S_2$, and $\lnot\hat{P}_1$ is defined to be the projector onto $S_1^{\bot}$. Then $\hat{P}_1\lor\hat{P}_2=\lnot(\lnot\hat{P}_1\land\lnot\hat{P}_2)$. If these operations are totally defined, then $P(\mathcal{H})$ is an orthocomplemented modular lattice, which presents several disadvantages such as not satisfying the distributive law \citep{Doering2010Topos}. Therefore, Kochen and Specker adopt $P(\mathcal{H})$ as partial Boolean algebra in 1960s\citep{Kochen1967The}. Define $\hat{P}_1\odot\hat{P}_2$ if $\hat{P}_1\hat{P}_2=\hat{P}_2\hat{P}_1$, that is, the compatibility relation. We have ${P}(\mathcal{H})=({P}(\mathcal{H}),\odot;\land,\lor,\lnot,\hat{0},\hat{1})$ is a partial Boolean algebra, where $\hat{0}$ is the zero projector, and $\hat{1}$ is the projector onto $\mathcal{H}$. A context ,or a Boolean subalgebra of $P(\mathcal{H})$, consists of mutually compatible projectors. Any ``quantum system" on $\mathcal{H}$ can be treated as partial Boolean subalgebra of $P(\mathcal{H})$.\par

\begin{definition}
A \textbf{quantum system} is defined by a partial Boolean subalgebra of $P(\mathcal{H})$ for some Hilbert space $\mathcal{H}$. We use $QS$ to denote the collection of all quantum systems.
\end{definition}

Quantum systems satisfy local consistency because the probability of $\hat{P}$ is $tr(\rho\hat{P})$, solely dependent on quantum state $\rho$ and independent of the contexts. Furthermore, if $\hat{P}_1\bot\hat{P}_2$, in other words, there is a projector $\hat{P}$ such that $\hat{P}_1\leq\hat{P}$ and $\hat{P}_2\leq\lnot\hat{P}$, then $\hat{P}_1,\hat{P}_2$ must be orthogonal, so $\hat{P}_1\odot\hat{P}_2$. Therefore, quantum systems satisfy exclusivity. In summary, we have

$$QS\subseteq epBA\subseteq pBA$$.

In practice, the primary consideration is finite $QS$. For example, five projectors on a 3-dimensional space are shown in Fig.\ref{fig2}.\par

\begin{figure}[h]
    \centering
    \includegraphics[width=0.4\linewidth]{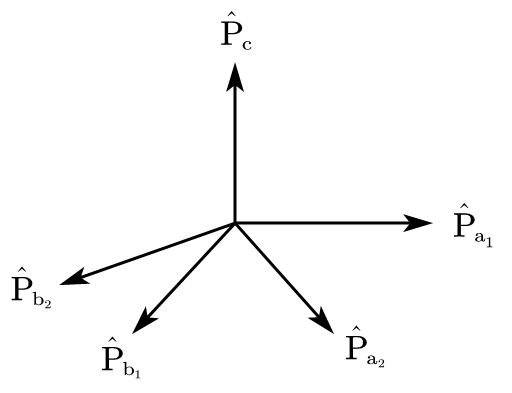}
    \caption{Five 3-dimensional rank-1 projectors. $\hat{P}_c,\hat{P}_{a_1},\hat{P}_{b_1}$ are pairwise orthogonal and $\hat{P}_c, \hat{P}_{a_2},\hat{P}_{b_2}$ are pairwise orthogonal}\label{fig2}
 \end{figure}

The five projectors generate $Q\in QS$, which is isomorphic to $B_1$ in Fig.\ref{pB1}. And $AG(B_1)\cong AG(Q)$, which is identical to the orthogonal graph of $\{\hat{P}_c,$ $\hat{P}_{a_1},$ $\hat{P}_{b_1},$ $\hat{P}_{a_2},$ $\hat{P}_{b_2}\}$. Any quantum experiment chooses a finite quantum system as its measurement scenario. The system of CHSH experiment for Bell nonlocality has 16 atoms\citep{Clauser1969Proposed}, and the system of Klyachko-Can-Binicioglu-Shumovsky (KCBS) experiment for quantum contextuality is generated by 5 atoms\citep{Alexander2008Simple}.\par

Since finite $QS\subseteq$ finite $epBA$, we obtain the proposition below from theorem \ref{thm1}.

\begin{proposition}\label{prop1}
If $Q_1$ and $Q_2$ are finite $QS$. then $Q_1\cong Q_2$ iff $AG(Q_1)\cong AG(Q_2)$.
\end{proposition}

On another hand, if $Q\in QS$, then a quantum state $\rho$ induces a map $\rho: Q\to[0,1]$, $\rho(\hat{P})=tr(\rho\hat{P})$. It is easy to prove that $\rho$ is a state on $Q$, therefore $\rho$ is determined by the probabilities of $A(Q)$, that is, the atom projectors of $Q$. Let $qs(Q)$ denote the states on $Q$ induced by quantum states. Derived from the theorem \ref{thm2},  we have

\begin{proposition}\label{prop2}
If $Q$ is a finite $QS$, then $qs(Q)\subseteq s(Q)\cong s(AG(Q))$
\end{proposition}

It means that any quantum state $\rho$ is determined by $\rho|_{A(Q)}$. The propositions \ref{prop1} and \ref{prop2} are the graph structure theorems for finite quantum systems. Therefore, the finite quantum systems are still determined by hidden-variables, but these hidden-variables need to be considered as a graph. We will generate these conclusions to more general quantum systems in the next section.\par

\section{The infinite cases}\label{sec_infinite}

The graph structure theorems also hold for the infinite quantum system on $\mathcal{H}$ of finite dimension. To prove it, we define the dimension of $epBA$.\par

\begin{definition}
If $B\in epBA$, $d(B)$ denotes the \textbf{dimension} of $B$. If $A(C)$ is finite for any maximal Boolean subalgebra $C$ of $B$, then $d(B):=\max\limits_{C}|A(C)|$ where $|A(C)|$ is the size of $A(C)$. Otherwise, $d(B):=\infty$.
\end{definition}

If $Q$ is a finite quantum system on $\mathcal{H}$ of infinite dimension, then there exists a $Q'\in QS$ on $\mathcal{H}'$ of finite dimension such that $Q'\cong Q$. Therefore, finite quantum systems are all finite dimensional $epBA$. Furthermore we have,

\begin{theorem}\label{thm3}
 If $B_1,B_2$ are finite dimensional $epBA$. then $B_1\cong B_2$ iff $AG(B_1)\cong AG(B_2)$.
\end{theorem}
\begin{theorem}\label{thm4}
If $B$ is a finite dimensional $epBA$, then $s(B)\cong s(AG(B))$.
\end{theorem}
The proof is given in Appendix \ref{A}.\par

If the dimension of $\mathcal{H}$ is finite, obviously the quantum systems $Q\in QS$ on $\mathcal{H}$ are finite dimensional $epBA$. Derived from the theorems \ref{thm3} and \ref{thm4}, we obtain the graph structure theorems for finite dimensional quantum systems as below.

\begin{proposition}\label{prop3}
If $Q_1,Q_2\in QS$ on $\mathcal{H}$ of finite dimension, then $Q_1\cong Q_2$ iff $AG(Q_1)\cong AG(Q_2)$.
\end{proposition}
\begin{proposition}\label{prop4}
If $Q\in QS$ on $\mathcal{H}$ of finite dimension, then $qs(Q)\subseteq$$s(Q)$$\cong s(AG(Q))$
\end{proposition}

For example, the Gleason's theorem\citep{Gleason1957Measures} indicates that the quantum states on $\mathcal{H}$ correspond to the functions on rank-1 projectors onto $\mathcal{H}$ when $3\leq d(\mathcal{H})<\infty$. Note that $P(\mathcal{H})\in QS$. We can restate the Gleason's theorem with the language of quantum system.

\begin{theorem}[Gleason\citep{Gleason1957Measures}]
If $3\leq d(\mathcal{H})<\infty$, then $qs(P(\mathcal{H}))=$$s(P(\mathcal{H}))$$\cong s(AG(P(\mathcal{H})))$.
\end{theorem}

For general $Q\in QS$, $qs(Q)=s(Q)$ may not hold, and usually $qs(Q)\subseteq s(Q)$. For example, if $\mathcal{H}$ is a 2-dimensional Hilbert space, $s(P(\mathcal{H}))$ contains $0-1$ states while $qs(P(\mathcal{H}))$ does not, so Gleason's theorem does not hold for $\mathcal{H}$ of dimension 2. Determining the $Q\in QS$ such that $qs(Q)=s(Q)$ is a meaningful question, because it indicates that $Q$ properly captures the features of quantum mechanics.\par

For the infinite dimensional $epBA$, we can generalize the theorem \ref{thm3}. If $B$ is a Boolean algebra, then $B$ is isomorphic to a power set algebra $\mathcal{P}(\Lambda)$ iff $B$ is atomic and complete. Therefore, atomic and complete Boolean algebra is the generalization of finite Boolean algebra. It is also true for partial Boolean algebra.\par

\begin{definition}
 Let $B\in pBA$. \par
 $B$ is said to be \textbf{atomic} if for each $x\in B$ and $x\neq 0$, there is $a\in A(B)$ such that $a\leq x$. \par
 $B$ is said to be \textbf{complete} if for each subset $S\subseteq B$ whose elements are pairwise compatible, $\bigvee S$ exists.\par
 Let $acepBA$ the atomic and complete $epBA$.
\end{definition}

The finite dimensional $B\in epBA$ is obviously atomic and complete. We have,

\begin{theorem}\label{thm5}
 If $B_1,B_2\in acepBA$. then $B_1\cong B_2$ iff $AG(B_1)\cong AG(B_2)$.
\end{theorem}
The proof is given in Appendix \ref{A}.\par

However, There is a challenge to achieve a generalization for the theorem \ref{thm4} in the infinite dimensional cases, because not all the events can be assigned probabilities. For example, the event $A\in\Delta$ has no probability if $\Delta\subseteq\mathbb{R}$ is not measurable. We discuss a potential method for the issue in the following text.\par

We relinquish the one-to-one correspondence between the states on atom graph and the system, and introduce a more abstract structure to characterize the infinite dimensional contextual systems which are determined by the atom graphs. For classical probability theory, the structure is a triplet $(\Lambda,B,p)$, where $B$ is a Boolean subalgebra of $\mathcal{P}(\Lambda)$ and $p$ is the probability measure. In the contextuality theory, we substitute a simple graph $G$ for $\Lambda$, and the algebra $B$ should be constructed by $G$. Therefore, we generalize the powerset operator $\mathcal{P}$ below.\par

\begin{definition}
 Let $G$ be a simple graph. If there is a $B\in acepBA$ such that $AG(B)=G$, then define $\mathcal{P}(G)=B$.
\end{definition}

Note that not any simple graph is the atom graph of an $acepBA$ (For instance, the 3-vertices graph $v_1-v_2-v_3$ cannot be any $pBA$'s atom graph), but if $G=AG(B)$, then $B\in acepBA$ is unique due to the theorem \ref{thm5}.

\begin{definition}
If $G$ is a simple graph, $\mathcal{P}(G)$ exists and $B$ is a sub $epBA$ of $\mathcal{P}(G)$, then $B$ is called a \textbf{measurable} $epBA$ on $G$. $G$ is called the atom graph of $B$.
\end{definition}

In fact, the measurable $epBA$ and $epBA$ are identical, and the only difference is that measurable $epBA$ is constructed by a given atom graph. If $p$ is a state on $B$,  then the triple $(G,B,p)$ forms a generalized probability theory for contextual systems with local consistency and exclusivity, where $B$ consists of the measurable events and $p$ is the probability measure on $B$.\par

\section{Exclusivity graph and atom graph}\label{sec_exgraph}

Given a set of events $\{E_i\}_{i=1}^n$, define a graph $G$ with $E_i$ as vertices and the exclusive relation as edges. $G$ is called an exclusivity graph\citep{Adan2010Contextuality,Adan2014Graph}. Exclusivity graph is an advanced method for quantum contextuality in finite cases, and have made considerable achievements\citep{Cabello2016Quantum,Xu2020Proof}. However, it does not reveal the graph structure of quantum mechanics due to two reasons. Firstly, the exclusivity graph is only applicable to rank-1 projectors, because the exclusivity may be not equivalent to compatibility for general projectors. Secondly, the exclusivity graph usually reflects a subsystem, which is significantly affected be the selection of vertices.\par

For example, the KCBS experiment \citep{Alexander2008Simple} considers five 3-dimensional rank-1 projectors $P_1=\{\hat{P}_0,\hat{P}_1,\hat{P}_2,\hat{P}_3,\hat{P}_4\}$ such that $\hat{P}_i$ and $\hat{P}_{i+1}$  (with the sum modulo 5) are orthogonal. Let $Q_{KCBS}\in QS$ denote the quantum system generated by $P_1$, $\hat{P}_{i\ i+1}$ denote $\lnot(\hat{P}_i\lor\hat{P}_{i+1})$, $i=0,1,2,3,4$, and $P_2=\{\hat{P}_{01},\hat{P}_{12},\hat{P}_{23},\hat{P}_{34},\hat{P}_{40}\}$.  One can verify that Fig.\ref{KCBS} shows the atom graph of $Q_{KCBS}$.

\begin{figure}[H]
    \centering
    \includegraphics[width=0.43\linewidth]{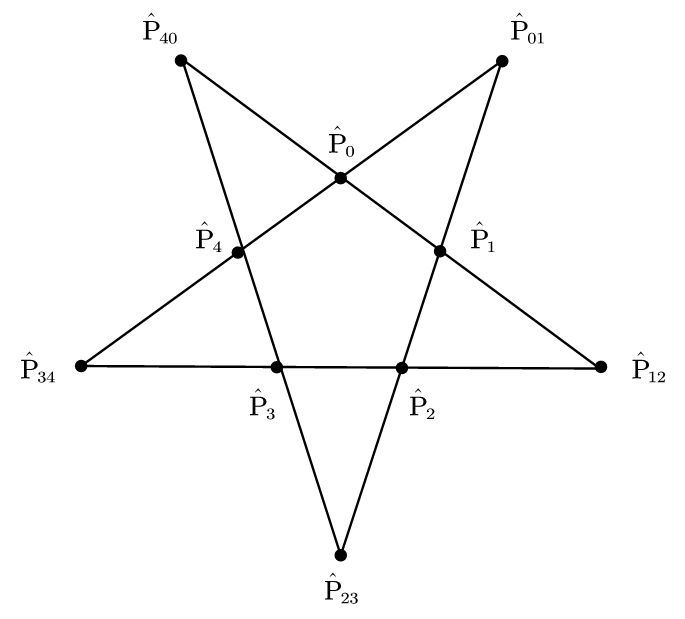}
    \caption{$AG(Q_{KCBS})$.}\label{KCBS}
 \end{figure}

For convenience, we introduce the notion "substate" for subgraph.

\begin{definition}
A \textbf{substate} on a graph $G$ is defined by a map $p:V(G)\rightarrow[0,\ 1]$ such that for each maximal clique $C$ of $G$, $\sum_{v\in C}p(v)\leq1$. Use $ss(G)$ to denote the set of substates on $G$.
\end{definition}

Consider two induced subgraphs of $AG(Q_{KCBS})$: $G_1=AG(Q_{KCBS})[P_1]$ and $G_2=AG(Q_{KCBS})[P_2]$, shown in Fig.\ref{G1} and \ref{G2}. Note that $P_1$ and $P_2$ both generate $Q_{KCBS}$, and $ss(G_1)\cong ss(G_2)\cong s(Q_{KCBS})$. For details, if $p\in s(Q_{KCBS})$, $p(\hat{P}_{i\  i+1})=1-p(\hat{P}_i)-p(\hat{P}_{i+1})$, $i=0,1,2,3,4$. Therefore, the substates $\{p|_{G_1}(\hat{P}_{i\ i+1})\}_{i=0}^4$ and $\{p|_{G_2}(\hat{P}_i)\}_{i=0}^4$ determine each other, then determine $p|_{AG(Q_{KCBS})}$. $Q_{KCBS}$ is a finite quantum system, so $p$ is determined. In other words, $G_1$ and $G_2$ both reconstruct the KCBS system.\par

\begin{figure}[H]
      \begin{minipage}[t]{0.5\linewidth}
          \centering
          \includegraphics[width=0.55\linewidth]{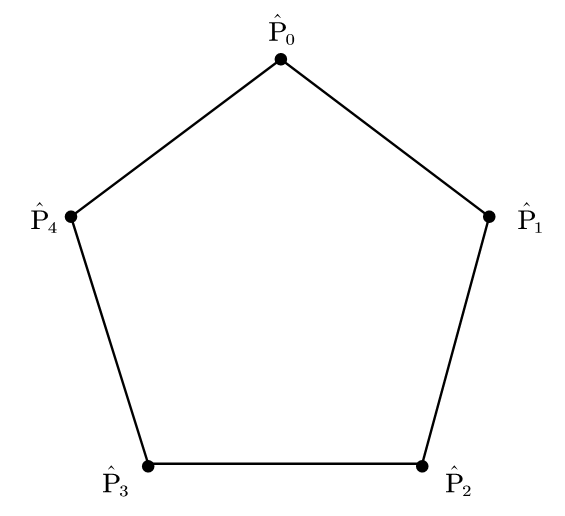}
          \caption{The exclusivity graph $G_1$}\label{G1}
      \end{minipage}
      \begin{minipage}[t]{0.5\linewidth}
          \includegraphics[width=0.55\linewidth]{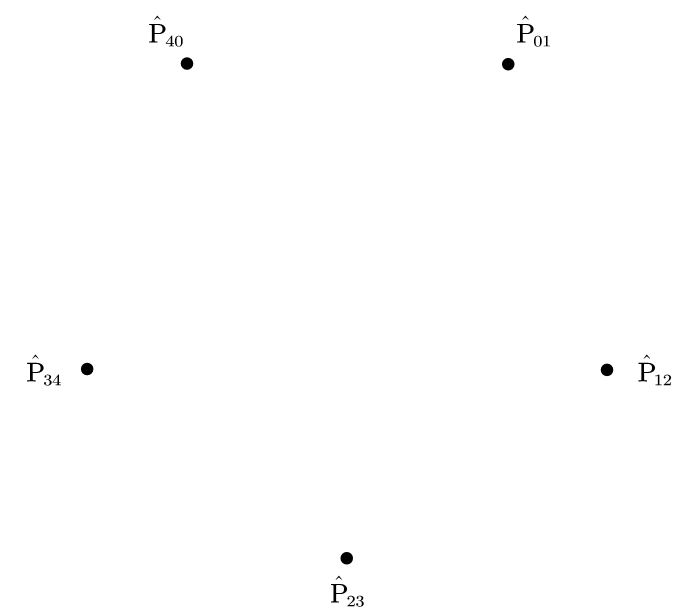}
          \caption{The exclusivity graph $G_2$}\label{G2}
      \end{minipage}
\end{figure}

However, some conclusions in exclusivity graph only hold for $G_1$ but not $G_2$, even though they depict the same system. Cabbelo\citep{Adan2010Contextuality,Adan2014Graph} points out that the weighted independence number $\alpha(G,w)$, weighted Lov\'{a}sz number $\vartheta(G,w)$ and weighted fractional packing number $\alpha^*(G,w)$ of the exclusivity graph $G$ can witness the quantum contextuality. For $G_1$, we have $$\alpha(G_1)=2<\vartheta(G_1)=\sqrt{5}<\alpha^*(G_1)=2.5,\ (w=1),$$ which means $G_1$ witnesses the violation of noncontextuality (NC) inequality $\sum_{v\in G_1}p(G_1)\leq\alpha(G_1)$. However, we have $$\alpha(G_2,w)=\vartheta(G_2,w)=\alpha^*(G_2,w)\ for\ any\ weight\ function\ w$$, which means any NC inequality $\sum_{v\in G_2}w_ip(G_2)\leq\alpha(G_2,w)$ on $G_2$ will not be violated, so $G_2$ does not present quantum contextuality. The reason of the problem is the exclusivity graph approach does not describe the complete system, in other words, exclusivity graph is the subgraph of atom graph, and the different selections of subgraphs may lose different properties of system.\par

With the conclusions of this paper, $epBA$ and atom graph provides a rigorous framework to uniquely and definitely describe the finite quantum systems. All the NC inequalities of a finite $Q\in QS$ have the form of $$\sum_{v\in AG(Q)}w(v)p(v)\leq\alpha(AG(Q),w),$$ where $w$ is a weight function $w: AG(Q)\to\mathbb{R}$, and $p\in s(AG(Q))$. If there exists a $w$ such that $\alpha(AG(Q),w)<\vartheta(AG(Q),w)$, then $Q$ presents contextuality. The selection of NC inequality, that is, the selection of $w$, is a significant field in quantum contextuality.

\section{Conclusion}

The quantum mechanics cannot be described by hidden-variable theory. However, a considerable portion of quantum systems possesses features analogous to hidden-variable theory. We prove that the finite dimensional quantum systems satisfy two graph structure theories (Proportions \ref{prop3} and \ref{prop4}), which imply that the atom graph is a rigorous mathematical characterization of finite dimensional quantum systems. We summarize our conclusions in Fig.\ref{collections}\par

\begin{figure}[h]
    \centering
    \includegraphics[width=0.9\linewidth]{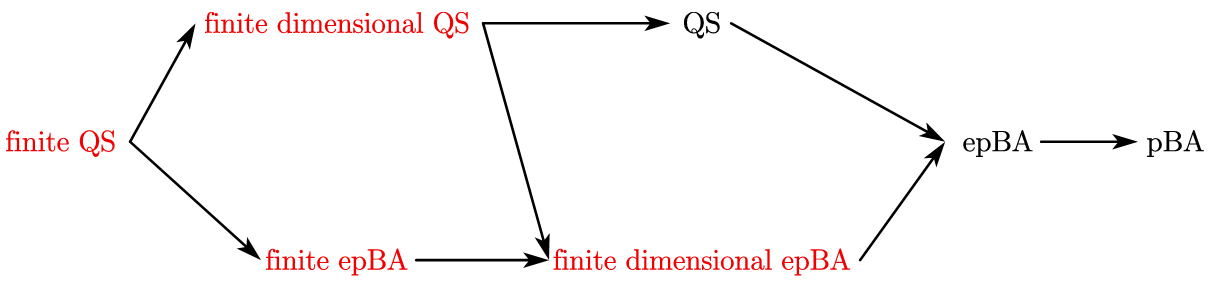}
    \caption{The inclusion relation of different collections. $A\to B$ means $A\subseteq B$. The red ones satisfy both two graph structure theorems. }\label{collections}
\end{figure}

For the infinite dimensional cases, we prove that the atomic and complete quantum systems have a one-one correspondence to their atom graphs(Theorem \ref{thm5}), and introduce a generalization from classical probability theory $(\Lambda,B,p)$ to the $(G,B,p)$ of graph $G$ and measurable $B\in epBA$ for contextual systems with local consistency and exclusivity.  The potential appliance of structure $(G,B,p)$ requires further research.\par

The atom graph formalizes the exclusivity graph approach, because each finite quantum system corresponds uniquely to an atom graph up to isomorphism. $epBA$ and atom graph have the potential to establish the mathematical framework of contextuality theory, to resolve some long-standing theoretical questions such as the proper definition of state-independent contextuality (SIC) and finding the minimal system with SIC\citep{Adan2023Kochen}.

\bmhead{Acknowledgments}
The work was supported by National Natural Science Foundation of China (Grant No. 12371016, 11871083) and National Key R\&D Program of China (Grant No. 2020YFE0204200).

\numberwithin{theorem}{section}
\appendix
\renewcommand{\thetheorem}{\thesection.\arabic{theorem}}

\section{Proof of the major theorems}\label{A}

We prove the theorems \ref{thm4} and \ref{thm5} here, which deduce the theorems \ref{thm1}, \ref{thm2} and \ref{thm3}.

\begin{lemma}\label{lem1_inf}
If $B\in acepBA$, and $C\subseteq B$ is a maximal Boolean subalgebra, then $A(C)\subseteq A(B)$.
\end{lemma}
\begin{proof}
Identical to the proof of lemma \ref{lem1_finite}.
\end{proof}
\begin{lemma}\label{lem2}
If $B\in acepBA$, $C_1,C_2$ are two maximal Boolean subalgebras of $B$, $A_1\subseteq A(C_1)$, $A_2\subseteq A(C_2)$, $A'_1=A(C_1)-A_1$ and $A'_2=A(C_2)-A_2$, then $\bigvee A_1=\bigvee A_2$ iff $A'_1\cup A_2=A(D_1)$ and $A_1\cup A'_2=A(D_2)$, where $D_1,D_2$ are maximal Boolean subalgebras of $B$.
\end{lemma}
\begin{proof}
Necessity: If $\bigvee A_1=\bigvee A_2=b\in B$, then $\bigvee A'_1=\bigvee A'_2=\neg b$. For any $a'_1\in A'_1$, $a'_1\leq\neg b$, and for any $a_2\in A_2$, $a_2\leq b$, so $a'_1\bot a_2$. Since $B$ is exclusive, $a'_1\odot a_2$. Therefore $A'_1\cup A_2$ is contained in a maximal Boolean subalgebra $D_1$, and $A'_1\cup A_2\subseteq A(D_1)$ due to the lemma \ref{lem1_inf}. Because $\bigvee (A'_1\cup A_2)=1$, $A'_1\cup A_2=A(D_1)$. Identically, $A_1\cup A'_2=A(D_2)$.\par

Sufficiency: If $A'_1\cup A_2=A(D_1)$ and $A_1\cup A'_2=A(D_2)$, we firstly prove that $A'_1\cap A_2=\emptyset$. Suppose there exists $a\in A'_1\cap A_2$, then $a\notin A_1$ and $a\notin A'_2$. On the other side, for all $a_1\in A_1$ and $a'_2\in A'_2$, $a\odot a_1$ and $a\odot a'_2$. Therefore $\{a\}\cup A_1 \cup A'_2$ is contained in a Boolean subalgebra, which contradicts that $D_2$ is maximal, so $A'_1\cap A_2=\emptyset$. Therefore $\bigvee A_1=\neg(\bigvee A'_1)=\bigvee A_2$.
\end{proof}

In the following proofs, similar to the lemma \ref{lem2}, if $A\subseteq A(C)$ for a maximal Boolean subalgebra $C$, then use $A'$ to denote $A(C)-A$.\par

\begin{theorem}[The theorem \ref{thm4}]
If $B$ is a finite dimensional $epBA$, then $s(B)\cong s(AG(B))$.
\end{theorem}
\begin{proof}
If $p\in s(B)$, obviously $p|_{A(B)}\in s(AG(B))$. Define $f:s(B)\to s(AG(B))$ as the restriction map below. We prove that $f$ is a bijection.

\begin{equation}
\begin{split}
f:s(B)&\to s(AG(B))\\
p&\mapsto p|_{A(B)}
\nonumber
\end{split}
\end{equation}

If $p_1,p_2\in s(B)$ and $p_1\neq p_2$, then $f(p_1)\neq f(p_2)$. Otherwise, suppose that $f(p_1)=f(p_2)=p_1|_{A(B)}=p_2|_{A(B)}$. For any $b\in B$, let $b\in C$ where C is a maximal Boolean subalgebra of $B$. Then $b=\bigvee A$, $A\subseteq A(C)\subseteq A(B)$ due to the lemma \ref{lem1_inf}. Thus $p_1(b)=p_1(\bigvee A)=\sum_{a\in A}p_1(a)=\sum_{a\in A}p_1|_{A(B)}(a)=\sum_{a\in A}p_2|_{A(B)}(a)=\sum_{a\in A}p_2(a)=p_2(\bigvee A)=p_2(b)$, so $p_1=p_2$, which induces a contradiction. Therefore $f$ is injective.\par

If $p'\in s(AG(B))$, then define $p:B\to[0,1]$ as below.
\begin{equation}
\begin{split}
p:B&\to [0,1]\\
0&\mapsto 0\\
b=\bigvee A&\mapsto\sum_{a\in A}p'(a),\ (A\subseteq A(B))
\nonumber
\end{split}
\end{equation}

Then $p|_{A(B)}=p'$, we prove that $p\in s(B)$. $p(b)=\sum_{a\in A}p'(a)\in[0,1]$ because $p'\in s(AG(B))$ and $A$ is contained in a maximal Boolean subalgebra. If $b=\bigvee A_1=\bigvee A_2$, then $A'_1\cup A_2=A(D_1)$ and $A'_1\cap A_2=\emptyset$ due to the proof of lemma \ref{lem2}, so $p(b)=\sum_{a\in A_1}p'(a)=1-\sum_{a\in A'_1}p'(a)=\sum_{a\in A_2}p'(a)$. Therefore, $p$ is well-defined.\par

We have $p(0)=0$. If $b=\bigvee A$, $p(\neg b)=p(\bigvee A')=\sum_{a\in A'}p'(a)=1-\sum_{a\in A}p'(a)=1-p(b)$. If $x,y\in B$ and $x\odot y$, then $x,\ y$ are in the same maximal Boolean subalgebra $C$. Let $x=\bigvee A_x,\ y=\bigvee A_y$ where $A_x, A_y\subseteq A(C)$. We have $p(x\lor y)+p(x\land y)=\sum_{a\in A_x\cup A_y}p'(a)+\sum_{a\in A_x\cap A_y}p'(a)=\sum_{a\in A_x}p'(a)+\sum_{a\in A_y}p'(a)=p(x)+p(y)$, so $p\in s(B)$. Therefore, $f(p)=p'$. $f$ is surjective.\par

In conclusion, $f$ is a bijection between $s(B)$ and $s(AG(B))$.
\end{proof}

\begin{theorem}[The theorem \ref{thm5}]
 If $B_1,B_2\in acepBA$. then $B_1\cong B_2$ iff $AG(B_1)\cong AG(B_2)$.
\end{theorem}
\begin{proof}
If $B_1\cong B_2$, $AG(B_1)\cong AG(B_2)$ obviously from the relevant definitions.\par

Conversely, if $g:A(B_1)\to A(B_2)$ is an isomorphism between $AG(B_1)$ and $AG(B_2)$, in other words, $a_1,a_2$ are adjacent iff $g(a_1),g(a_2)$ are adjacent, we define a map from $B_1$ to $B_2$ as follow.
\begin{equation}
\begin{split}
f:B_1&\to B_2\\
0 &\mapsto 0\\
b=\bigvee A_1&\mapsto\bigvee g(A_1).\ (A_1\subseteq A(B_1))
\nonumber
\end{split}
\end{equation}

The $C,C_1$ and $C_2$ below are all maximal Boolean subalgebras of $B_1$.

$f(b)=\bigvee g(A_1)$ exists because $B_2$ is complete. Suppose $A_1\subseteq A(C_1)$ and $A_1\subseteq A(C_2)$. If $b=\bigvee A_1=\bigvee A_2$, then $A'_1\cup A_2=A(D_1)$ and $A_1\cup A'_2=A(D_2)$ due to the proportion $1$ in lemma \ref{lem2}, so $A'_1\cup A_2$ and $A_1\cup A'_2$ are both maximal cliques of $AG(B_1)$. Suppose that $\bigvee g(A_1)\neq\bigvee g(A_2)$. Because of the proportion $2$ in lemma \ref{lem2}, one of $g(A'_1\cup A_2)$ and $g(A_1\cup A'_2)$ is not a maximal clique of $AG(B_2)$, which contradicts that $g$ is an isomorphism between $AG(B_1)$ and $AG(B_2)$. Therefore, $f$ is well-defined. \par

If $b_1=\bigvee A^1_{1}\in B_1$, $b_2=\bigvee A^2_{1}\in B_1$ and $b_1\neq b_2$, then $f(b_1)=\bigvee g(A^1_1)$, $f(b_2)=\bigvee g(A^2_1)$. Similarly, $f(b_1)\neq f(b_2)$ because of the lemma \ref{lem2} and the isomorphism $g$, so $f$ is injective. For any $b=\bigvee A\in B_2$, $f(\bigvee g^{-1}(A))=b$, so $f$ is surjective. Therefore, $f$ is a bijection.\par

Finally, $f(0)=0$. For $b=\bigvee A\in B_1$, suppose that $A\subseteq C$. $f(\neg b)=f(\neg(\bigvee A))=f(\bigvee A')=\bigvee g(A')=\neg\bigvee g(A)=\neg f(b)$. If $b_1, b_2\in B_1$ and $b_1\odot b_2$, then let $b_1=\bigvee A_1,b_2=\bigvee A_2\in C$, so $A_1,A_2\subseteq A(C)$, $f(b_1),f(b_2)\in f(C)$ and $f(b_1)\odot f(b_2)$. Furthermore, $f(b_1\lor b_2)=f(\bigvee A_1\lor\bigvee A_2)=f(\bigvee(A_1\cup A_2))=\bigvee g(A_1\cup A_2)=\bigvee g(A_1)\lor\bigvee g(A_2)=f(b_1)\lor f(b_2)$. Therefore, $f$ is a homomorphism.\par

In conclusion, $f$ is an isomorphism between $B_1$ and $B_2$.
\end{proof}

\bibliography{sn-bibliography}

\end{document}